\documentclass[12, reqno] {amsart}
\usepackage{amsmath}
\usepackage{amscd,amsthm,amsfonts,amsopn,amssymb,verbatim}
\numberwithin{equation}{section}
\usepackage{epsfig}
\usepackage{graphicx,epsfig}
\usepackage{amsmath,amsthm,amssymb,graphicx}
\usepackage{color}
\usepackage{times}

\newtheorem{theorem}{Theorem} 

\newtheorem{corollary}[theorem]{Corollary}

\newtheorem{lemma}[theorem]{Lemma}

\theoremstyle{remark}

\def\be{\begin{equation}}
\def\ee{\end{equation}}

\def\ve{\varepsilon}

\allowdisplaybreaks

\begin{document}
\setlength{\parskip}{2pt}

\title 
[A lower bound for the Lyapounov exponents]
{A lower bound for the Lyapounov exponents of the random Schr\"odinger operator on a strip}
\author{J.Bourgain}

\date{\today}

\begin{abstract}
We consider the random Schr\"odinger operator on a strip of width $W$, 
assuming the site distribution of bounded density.
It is shown that the positive Lyapounov exponents satisfy a 
lower bound roughly exponential in $-W$ for $W\to\infty$.
The argument proceeds directly by establishing Green's function decay, but does not appeal to Furstenberg's random matrix theory on the strip.
One ingredient involved is the construction of `barriers' using the RSO theory on $\mathbb Z$.
\end{abstract}
\maketitle

\section
{Introduction}

We consider the classical one-dimensional Anderson model on a strip of width $W$, thus
\be\label{0.1}
H=\lambda V+\Delta
\ee
with $\Delta$, the lattice Laplacian on $\mathbb Z\times\mathbb Z_W$, $\mathbb Z_W=\mathbb Z/W\mathbb Z$
(periodic boundary conditions) and $V=(V_{ij})_{i\in\mathbb Z, j\in\mathbb Z_W}$
 a random potential with IID site distribution.
It is well-known that for any $\lambda\not= 0$, this model exhibits Anderson localization.
The non-perturbative approach is provided by Furstenberg's random matrix product theory, applied to the underlying
transfer operators in the symplectic group $Sp(2W)$; cf. \cite{B-L}.
The argument is non-quantitative, in the sense that no explicit lower bounds on the $W$ positive Lyapounov exponents
is provided.
Hence our concern in this Note is to obtain a lower bound in terms of $W$.
If $\lambda$ is taken sufficiently small (depending on $W$) in \eqref {0.1}, a very explicit analysis based on an extension 
of the Figotin-Pastur method appears in \cite {S-B}, leading to exact formulas for the Lyapounov exponents.
Unfortunately, this technique seems restricted to the perturbative setting.
Related work for random band matrices in \cite {S} leads to upper bounds on the localization length of the form
$W^C$ (the conjecture in this setting is a localization length $O(W^2)$, which seems unproven at this point).
Of course, the Schr\"odinger model \eqref {0.1} is much `sparser' and there does not appear to be an easy way to adjust
the technique from \cite{S} to our setting.
We settle here for the modest goal of establishing an explicit upper bound on the localization length for random $SO$ on
a $W$-strip, assuming for simplicity that the site distribution of the potential has a bounded density.
It is possible to adjust the argument to treat other (continuous) densities, but definitively the Bernoulli model is not
captured (mainly due to the lack of a quantitative Wegner estimate in the Bernoulli-setting).
Our estimate is roughly exponential in $W$ (while one could again conjecture that a powerlike behavior is the true answer).
We will not use Furstenberg theory, except for $W=1$ (see Lemma 2, which is a crucial ingredient).

We refer the reader in particular \cite{K-L-S1}, \cite{K-L-S2} for treatments of
localization and density of states for the Anderson model on the strip and
\cite{B-L} as reference work.

\noindent
{\bf Acknowledgement.} The author is grateful to A.~Klein for several
stimulating discussions on the issue discussed in this paper. 
\bigskip

\section
{Use of the Shur complement formula}

In what follows, we make essential use of the following principle
\begin{lemma}\label{Lemma1}
Let $T$ be selfadjoint with finite index set $\Omega$.

Let $\Omega=\Omega_1\cup\Omega_2$ be a decomposition and set $T_i={\Omega_i}TR_{\Omega_i} (i=1, 2)$.
Assume $T_2$ invertible.

Let $D_V$ be the diagonal operator defined by
$$
D_V=\sum_{i\in\Omega_1} V_i e_i\otimes e_i
$$
with $V_i\in\mathbb R$ IID with bounded density distribution. Denote
\be\label{1.1}
T_V=D_V +T.  \ee 
Then 
\be\label{1.2} \mathbb P_V[\Vert R_{\Omega_1} T_V^{-1} R_{\Omega_1}\Vert >\lambda] \lesssim |\Omega_1| \lambda^{-1}.  
\ee 
\end{lemma}

\begin{proof} By the Shur complement formula 
\begin{align}\label{1.3} R_{\Omega_1}T_V^{-1} R_{\Omega_1}&= (D_V+T_1 - R_{\Omega_1} TR_{\Omega_2} T_2^{-1} R_{\Omega_2} TR_{\Omega_1})^{-1}
\nonumber\\ 
& = (D_V+A)^{-1} 
\end{align} 
and 
\be\label{1.4} 
\mathbb P_V[\text{dist\,} \big(\sigma(D_V+A), 0\big) <\kappa]\lesssim \kappa|\Omega_1|.  
\ee 
The claim follows.  
\end{proof} 

\bigskip 
\section {Construction of barriers} 

Let $W\geq 1$ be an integer and consider SO of the form
$$ 
H= V+ \Delta 
$$ 
on the band $\mathbb Z\times\mathbb Z_W$, $\mathbb Z_W=\mathbb Z/{W\mathbb Z}$ (i.e. periodic $bc$) with $\Delta$ the nearest neighbor 
Laplacian on $\mathbb Z\times\mathbb Z_W$ and $V$ a random potential 
$V= (V_{ij})_ {i\in\mathbb Z, j\in\mathbb Z_W}, V_{ij} $ IID.
If $I\subset\mathbb Z$ is an interval, $H_I$ denotes the corresponding restriction of $H$.  

\begin{lemma}\label{Lemma2} 
Let $I$ be an interval of size 
\be\label{2.1} 
N>C[\log (1+W)]^2.  
\ee 
Fix an energy $E$.  
Then, with above notations, the properties \be\label{2.2}
 \Vert(H_I-E)^{-1}\Vert < e^{\sqrt N} 
\ee
 and 
\be\label{2.3} |(H_I-E)^{-1} \big((i, j), (i', j')\big)|< e^{-cN} \text { for } i, i'\in I, |i-i'|>\frac N{10} 
\text { and } j, j'\in \mathbb Z_W 
\ee 
hold with probability at least $C^{-N^2W}$.  
\end{lemma} 

This statement is also valid in the Bernoulli case.  

\begin{proof} 
The main idea is to deduce the statement from the case $W=1$.  Let $I=[0, N-1]$.  Let $(v_i)_{i\in I}=v$ be assignments of the potential and set 
\be\label{2.4} 
V_{ij}=v_i \text { for } i\in I, j\in\mathbb Z_W.  
\ee 
Considering the SO $h$ on $\mathbb Z$ with potential $(V_i)_{i\in\mathbb Z}$, for any given energy $E'\in\mathbb R$, the restricted Green's function $(h_I-E')^{-1}$ will satisfy bounds 
\be\label{2.5} 
\Vert (h_I-E')^{-1} \Vert < e^{\sqrt N} 
\ee 
and 
\be\label{2.6} 
|(h_I-E')^{-1} (i, i')|<e^{-cN}\text { for } |i-i'|>\frac N{10} 
\ee 
excluding a set of $(v_i)_{i\in I}$ of measure at most $e^{-c\sqrt N}$.  
We assume here $N$ sufficiently large.  
The latter statement follows from the transfer matrix approach and is equally valid for Bernoulli-distributions.  

Consider next the equation 
\be\label{2.7} 
(H_I-E)\xi =\eta 
\ee 
with $\xi =\sum_{i\in I} \xi_ie_i, \eta =\sum_{i\in I}\eta_i e_i$ and $V$ satisfying \eqref{2.4}.  Thus 
\be\label {2.8} 
(v_i-E)\xi_{ij} +\xi_{i-1, j}+\xi_{i+1, j}+\xi_{i, j-1}+\xi_{i, j+1} =\eta_{i, j} \text { for } i\in I, j\in\mathbb Z_W 
\ee 
and Dirichlet bc in $i$.  

Denote $e(\theta)= e^{2\pi i\theta}$.  
Define for $\theta\in \{\frac wW; 0\leq w< W\}$ 
$$ 
\hat\xi_i(\theta) =\sum_{j\in\mathbb Z_W} e(j\theta)\xi_{i, j} 
$$ 
and similarly $\hat\eta_i(\theta)$.  
It follows thus from \eqref{2.8} that 
$$ 
(V_i-E) \hat\xi_i\theta)+\hat\xi_{i-1} (\theta) +\hat\xi_{i+1}(\theta)+2\cos 2\pi\theta \ \hat{\xi_i}(\theta) 
=\hat\eta_i(\theta) \text { for } i\in I.  
$$ 
Hence 
\be\label{2.9} 
(h_I-E') \hat\xi(\theta)=\hat\eta(\theta) 
\ee 
with $E'=E-2\cos 2\pi \theta$.  

We choose $(v_i)_{i\in I}$ in \eqref{2.4} as to ensure \eqref {2.5}, \eqref {2.6} for \hfill\break
$E'\in E+\{2\cos2\pi\theta; \theta\in \mathbb Z_W\}$.  
This holds indeed with large probability in $v$, if we assume 
\be\label{2.10} 
N> C(\log W)^2.  
\ee 
We verify properties \eqref{2.2} and \eqref{2.3}.  

Let $\Vert\eta\Vert =1$ in \eqref{2.7}.  
It follows from \eqref{2.5}, \eqref{2.9} that for $\theta \in\mathbb Z_W$ 
$$ \Vert \hat\xi(\theta)\Vert \leq \Vert (h_I-E')^{-1}\Vert \ 
\Vert \hat\eta(\theta) \Vert< e^{\sqrt N}\Vert\hat\eta(\theta)\Vert.  
$$ 
Squaring both sides and averaging over $\theta\in\mathbb Z_W$ implies by Parseval that 
$$ 
\Vert\xi\Vert^2 < e^{2\sqrt N} \Vert\eta\Vert^2, \Vert\xi\Vert< e^{\sqrt N} $$ hence \eqref{2.2}.  

Next, take $\eta= e_{i, j}, \xi =(H_I-E)^{-1}\eta$. 
Thus 
$$ \xi_{i', j'} =\langle (H_I-E)^{-1} \eta, e_{i', j'}\rangle.  
$$ 
Again by \eqref{2.9}, for each $\theta\in\mathbb Z_W$ 
$$ \begin{aligned} 
|\hat\xi_{i'} (\theta)| &=|\langle (h_I-E')^{-1} \hat\eta (\theta), e_{i'}\rangle|\\ 
&\leq |(h_I-E')^{-1} (i, i')|< e^{-cN} \end{aligned}.
$$ 
Therefore clearly 
$$ 
|\xi_{i', j'}|< e^{-cN} 
$$ 
proving \eqref{2.3}.  

Recall that $V=(V_{ij})_{i\in I, j\in\mathbb Z_W}$ was taken to satisfy \eqref{2.4}, $(v_i)_{i\in I}$ 
taken in a set of measure at least $\frac 12$.  Clearly \eqref{2.2} and elementary perturbation theory shows that 
assumption \eqref{2.4} may be weakened to 
\be\label{2.11} |V_{ij}-v_i|< e^{-N} \text { for } i\in I, j\in\mathbb Z_W 
\ee 
and this property will hold with measure at least $C^{-N^2W}$.  
Lemma \ref{Lemma2} follows.  
\end{proof} 
\bigskip 

\section {Restricted Green's function estimates} 

Let $H$ be as in \S2 and assume the potential distribution with bounded density for simplicity.  

Fix $E$ and denote $G_I=(H_I-E)^{-1}$.  
The basic construction proceeds as follows.  

Fix $M>(\log W)^2$ and let 
\be\label{3.1} 
N> C^{M^2W} 
\ee 
be a multiple of $M$.  

Consider the intervals $I_\alpha= ]\alpha M, (\alpha+1) M[ \subset I=[0, N]$. 
 \vskip.2 true in 

\centerline{
\begin{minipage}{2.5 in}
\centerline
{\hbox{\includegraphics[width = 5.5in]{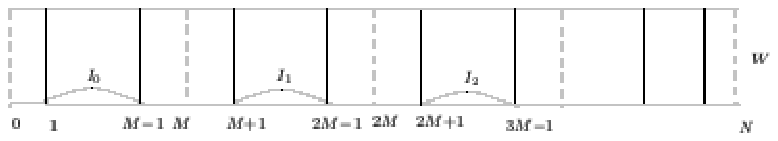}}}
\end{minipage}}


 \vskip.2 true in 

Say that $\alpha$ is good provided 
\be\label{3.2} 
\Vert G_{I_\alpha}\Vert< e^{\sqrt M} 
\ee 
and 
\be\label{3.3} 
\Vert P_{\{\alpha M+1\}} G_{I_\alpha} P_{\{(\alpha+1)M-1\}}\Vert < e^{-cM} 
\ee 
with $P_{\{i\}}$ the projection on $[e_{i, j}; j\in\mathbb Z_W]$.

According to Lemma \ref{Lemma2}, $\alpha$ will be good with probability at least $C^{-M^2W}$.
Note that this event only depends on the variables $(V_{i, j})_{i\in I_\alpha , j\in\mathbb Z_W}$.
Hence, by our choice of $N$, there will be at least $R=[e^{\frac c{10} M}]$ good $\alpha$'s with probability $> 1-e^{-\sqrt N}$.
This statement only involves the variables 
$ (V_{i, j})_{\stackrel {i\not= 0 (\text{mod\,} M)}{j\in\mathbb Z_W}}$ 
which we {\it fix}.
Denote $I_1, \ldots, I_R, I_r=] k_rM, (k_1+1)M[$
these good intervals, that will be used as barriers.

From the resolvent identity
\be\label{3.4}
\Vert P_{\{0\}} G_{[0, N]} P_{\{N\}}\Vert\leq
\Vert P_{\{0\}} G_{[0, (k_1+1)M-1]} P_{\{ (k_1+1)M-1\}} \Vert \ \Vert P_{\{(k_1+1))M\}} G_{[0, N]} P_{\{N\}}\Vert
\ee
and again by the resolvent identity and \eqref{3.3},
the first factor on the rhs of \eqref{3.4} may be bounded by
\begin{align}\label{3.5}
&\Vert P_{\{0\}} G_{[0, (k_1+1)M-1]} P_{\{k_1 M\}}\Vert \ \Vert P_{\{k_1 M+1\}} G_{I_1} P_{\{(k_1+1)M-1\}}\Vert\nonumber\\
& < e^{-cM}\Vert P_{\{0\}} G_{[0, (k_1+1)M-1]} P_{\{k_1 M\}}\Vert
\end{align}

The factor
\be\label{3.6}
\Vert P_{\{0\}} G_{[0, (k_1+1)M-1]} P_{\{k_1M\}}\Vert
\ee
will be bounded by the Shur complement formula, exploiting the variables
\be\label{3.7}
(V_{ij}) _{\stackrel{i=0, k_1M}{j\in\mathbb Z_W}}.
\ee

We apply Lemma \ref{Lemma1} with $\Omega=[0, (k_1+1)M[\times \mathbb Z_W$ and $\Omega_1=\{0, k_1M\}\times \mathbb Z_W$.
Hence, by \eqref{1.2}, we may ensure that
\be\label{3.8}
\eqref{3.6} \leq \Vert P_{\Omega_1} G_{[0, (k_1+1)M-1]} P_{\Omega_1}\Vert < e^{\frac c2 M}
\ee
excluding a set of measure at most $e^{-\frac c3M}$ in $(V_{i, j})_{i\equiv 0(\mod M)}$.

From \eqref{3.4}, \eqref{3.5}, \eqref{3.8}, we obtain then
\be\label{3.9}
\Vert P_{\{0\}} G_{[0, n]} P_{\{N\}} \Vert< e^{-\frac c2M} \Vert P_{\{ (k_1+1)M\}} G_{[0, N]} P_{\{N\}}\Vert.
\ee
Repeating the argument considering the next barrier $I_2$, which
$$
\Vert P_{\{(k_1+1)M\}} G_{[0, N]} P_{\{N\}}\Vert \leq
\Vert P_{\{(k_1+1)M\}} G_{[0, (k_2+1)M-1]} P_{\{(k_2+1)M-1\}} \Vert \ \Vert P_{\{(k_2+1)M\}} G_{[0, n]} P_{\{N\}}\Vert
$$
and
$$
\Vert P_{\{(k_1+1)M\}} G_{[0, (k_2+1)M-1]} P_{\{(k_1+1)M\}}\Vert \leq
e^{-cM} \Vert P_{\{(k_1+1)M\}} G_{[0, (k_2+1)M-1]} P_{\{ k_2M\}}\Vert
$$
etc.

For the last factor $\Vert P_{\{(k_R+1)M\}} G_{[0, n]} P_{\{N\}}\Vert$, apply again Lemma \ref{Lemma1} in order
to get a bound by $e^M$.

The above iteration shows that we may estimate
\be\label{3.10}
\Vert P_{\{0\}} G_{[0, N]} P_{\{N\}}\Vert < e^{-\frac c2RM}
\ee
by exclusion in the $(V_{i, j})_{\stackrel {i\equiv 0 (\text{mod\,} M)}{j\in \mathbb Z_W}} $ variable of a set of measure at most
\be\label{3.11}
(R+1) e^{-\frac c3M}< e^{-\frac c5 M}
\ee
by our choice of $R$.

Taking $M > (\log W)^2, N=C^{M^2W}$, we proved the following.

\begin{lemma}\label{Lemma3}
Let $H$ be as in \S2 with potential distribution of bounded density. Let
\be\label{3.12} 
N > C^{W(\log W)^4}
\ee
and $I\subset\mathbb Z$ an interval of size $N, I=[a, b]$.

Fix $E$. Then,
\be\label{3.13}
\Vert P_a(H_I-E)^{-1} P_b\Vert< e^{-\exp(\frac {\log N}W)^{\frac 12}}
\ee
outside an exceptional  set of measure  at most $C e^{-c(\frac{\log N} W)^{\frac 12}}$ 
\end{lemma}

Starting from this statement, we perform the usual multi-scale analysis.

We use the following bootstrap lemma.

\begin{lemma}\label{Lemma4}
Let $H$ be as in Lemma \ref{Lemma3} and fix $E$.
Let $M$ be a scale, \hbox{$0<\ve, \delta<1$,} such that
\be\label{3.14}
\Vert P_a G_IP_b\Vert < e^{-\delta M}
\ee
if $I=[a, b] \subset \mathbb Z$ is an interval of size $M-1$, $G_I=(H_I-E)^{-1}$, hold with probability at least
$1-\ve$ in $V$.

Let $r\in \mathbb Z_+$ and assume further that
\be\label{3.15}
W+\frac 1\ve < c \, e^{\frac{\delta M}{4r}}
\ee
($c>0$ a constant depending on the density of the site distribution).

Take $n\in\mathbb Z_+$ such that
\be\label{3.16}
W+\frac 1\ve < n< c \, e^{\frac {\delta M}{4r}}
\ee
and set $N=n.M$.
Then \eqref{3.14} will hold at scale $N+1$ with $\ve, \delta$ replaced by
\begin{align}
\label{3.17}
\ve_1 &= 2^{-\sqrt n}+ e^{-\frac {\delta M}{2r}} \\
\delta_1&= (1-\sqrt \ve) \Big(1-\frac 1r\Big)\delta\label{3.18}
\end{align}
\end{lemma}

\begin{proof}

We make the same construction as in the proof of Lemma \ref{Lemma3} earlier in this section, using
the same notation.
Say that $\alpha$ is `good' if $I=I_\alpha$ satisfies \eqref{3.14}.
Denote $I_1, \ldots, I_R$ the good $I_\alpha$-intervals, which only depend on the variables
$(V_{ij})_{\substack{i\not=0 (\text{mod\,} M) \\ j\in\mathbb Z_W}}$.

Since $\alpha$ is good with probability at least $1-\ve$, it follows that $R>(1-\sqrt \ve)n$ with probability at
least $1-e^{-\sqrt\ve n}$.
Proceeding exactly as in the proof of Lemma \ref{Lemma3}, we repeat the same iteration.
Thus we write \eqref{3.4}, \eqref{3.5} with in \eqref{3.5} the factor $e^{-cM}$ replaced by $e^{-\delta M}$.
The factor \eqref{3.6} is again bounded using Lemma \ref{Lemma1}, requiring this time that \eqref{3.6} is bounded
by $e^{\frac {\delta M}r}$, which will hold with probability at least $1-CW e^{-\frac {\delta M}r}$.
An $R$-fold iteration gives instead of \eqref{3.10} that
\be\label{3.19}
\Vert P_{\{0\}} G_{[0, N]} P_{\{N\}}\Vert < e^{-R(1-\frac 1r)\delta M}< e^{-(1-\sqrt \ve)(1-\frac 1r){\delta N}}
\ee
which by the preceding will hold outside an exceptional set of measure at most
\be\label{3.20}
e^{-\sqrt \ve n}+CW n e^{-\frac {\delta M}{2r}}< 2^{-\sqrt n}+ e^{-\frac {\delta M}{2r}}
\ee 
in view of assumption \eqref{3.16}.
This proves the lemma.
\end{proof}

Returning to Lemma \ref{Lemma3}, set
\be\label{3.21}
N_0=A^{W(\log W)^4}
\ee
with $A$ a sufficiently large constant (independent of $W$) and
\begin{align}\label{3.22}
\delta_0&=\frac 1{N_0} \exp \Big(\frac {\log N_0}W\Big)^{\frac 12}\\
\ve_0& =\exp \Big(-c\Big(\frac{\log N_0}W\Big)^{\frac 12}\Big).\label{3.23}
\end{align}
Thus \eqref{3.14} holds with probability at least $1-\ve_0$.
Take $r=10, n=N_0$.
Condition \eqref{3.16} will clearly hold for $A$ large enough.
According to Lemma \ref{Lemma4}, $N_1\sim n N_0=N_0^2$ will satisfy \eqref{3.14}, where, by \eqref{3.17}, \eqref{3.18}, we can take
$$
\ve_1= \frac 1{N_1} \text { and }  \delta_1= (1-\sqrt{\ve_0})\Big(1-\frac 1{10}\Big)\delta_0.
$$
A further iteration based on Lemma \ref{Lemma4} easily leads to 
$$
\begin{aligned}
N_{s+1} & =N^2_s\\
\ve_s&= \frac 1{N_s}\\
\delta_{s+1} &=(1-\sqrt{\ve_s}) (1-\frac 1{10^s}) \delta_s> \frac 12\delta_0.
\end{aligned}
$$
We obtain therefore the following amplification of Lemma \ref{Lemma3}.

\begin{lemma}\label{Lemma5}
Under the assumptions of Lemma \ref{Lemma3}, for fixed $E$ and $N>C^{W(\log W)^4}$,
\be\label{3.24}
\Vert P_a (H_I-E)^{-1} P_b\Vert < e^{-\frac 12\delta_0 N}
\ee
with
\be\label{3.25}
\delta_0 = C^{-W(\log W)^4}
\ee
holds for $I=[a, b] \subset\mathbb Z$ an $N$-interval, outside an exceptional set of measure at most $e^{-\delta_0N^{1/3}}$.
\end{lemma}

In particular, this yields.

\begin{corollary}\label{Corollary6}
Let $H$ be a random $SO$ on a strip of width $W$ and site distribution of bounded density.
Then its positive Lyapounov exponents are at least
$$
C^{- W(\log W)^4}.
$$
\end{corollary}

\end{document}